\documentclass[10pt, conference, compsocconf]{IEEEtran}

\usepackage[cmex10]{amsmath}

\usepackage{amssymb,amsfonts,textcomp,amsthm}
\usepackage{array}
\usepackage{delarray}
\usepackage{latexsym}

\usepackage[pdftex]{color, graphicx}
\usepackage{subfigure}
\usepackage{mathptmx}
\usepackage{times}
\usepackage{paralist}
\usepackage{tabularx}
\usepackage{multirow}
\usepackage{graphicx}
\usepackage{url}
\usepackage{rotating}

\usepackage[lined, ruled, commentsnumbered]{algorithm2e}
\usepackage[font=footnotesize,labelfont=bf]{caption}

\DeclareMathOperator \dir {dir}

\usepackage{cite}

\newtheorem{proposition}{Proposition}
\newtheorem{theorem}{Theorem}

\newtheorem{lemma}{Lemma}

\newenvironment{proof-sketch}{\noindent \textit{Sketch of Proof:}}{$\Box$}

\title{Network delay-aware load balancing in selfish and cooperative distributed systems}

\author{
\IEEEauthorblockN{Piotr
  Skowron
}
\IEEEauthorblockA{Faculty of Mathematics, Informatics and Mechanics\\
  University of Warsaw \\
  Email: p.skowron@mimuw.edu.pl
}
\and
\IEEEauthorblockN{Krzysztof
  Rzadca
}
\IEEEauthorblockA{Faculty of Mathematics, Informatics and Mechanics\\
  University of Warsaw \\
  Email: krzadca@mimuw.edu.pl
}

}

\begin{document}
\maketitle

\begin{abstract}
  We consider a request processing system composed of organizations and their servers connected by the Internet.
  The latency a user observes is a sum of communication delays and the time needed to handle the request on a server. The handling time depends on the server congestion, i.e. the total number of requests a server must handle. We analyze the problem of balancing the load in a network of servers in order to minimize the total observed latency. We consider both cooperative and selfish organizations (each organization aiming to minimize the latency of the locally-produced requests). The problem can be generalized to the task scheduling in a distributed cloud; or to content delivery in an organizationally-distributed CDNs.

  In a cooperative network, we show that the problem is polynomially solvable. We also present a distributed algorithm iteratively balancing the load. We show how to estimate the distance between the current solution and the optimum based on the amount of load exchanged by the algorithm. During the experimental evaluation, we show that the distributed algorithm is efficient, therefore it can be used in networks with dynamically changing loads.

  In a network of selfish organizations, we prove that the price of anarchy (the worst-case loss of performance due to selfishness) is low when the network is homogeneous and the servers are loaded (the request handling time is high compared to the communication delay). After relaxing these assumptions, we assess the loss of performance caused by the selfishness experimentally, showing that it remains low.

  Our results indicate that a network of servers handling requests can be efficiently managed by a distributed algorithm. Additionally, even if the network is organizationally distributed, with individual organizations optimizing performance of their requests, the network remains efficient.

\end{abstract}

\vspace{-7pt}
\section{Introduction}\label{sec::introduction}

One of the most important aspects affecting the perceived quality of a web service is the delay in accessing the content. To avoid servers' congestion, the content of the web pages is commonly replicated in multiple locations. Additionally, in order to minimize the network latency, the replicas are placed close to users. Because the intensity of the web traffic changes dynamically, efficient mirroring requires both expensive infrastructure and effective load balancing algorithms. As the result many organizations decide to handle the task of mirroring their data to dedicated platforms --- content delivery networks (CDNs)~\cite{pallis2006content, freedman2010experiences}. The CDNs have been very successful in the recent years: Akamai~\cite{Nygren:2010:ANP:1842733.1842736, akamai2, draftingBehindAkamai}, the largest CDN, handles around 15-20\% of the Internet traffic.

Consider an apparently different distributed system: a cloud of datacenters performing com\-pu\-ta\-tio\-na\-lly-intensive parallel calculations. Each datacenter attempts to accelerate its calculations by distributing some of its load to less loaded and faster datacenters. However, the datacenters in remote locations must be avoided as the time needed to transfer the input and the result may dominate the processing time.

Routing the requests in a CDN and distributing the load in a cloud are strongly related problems. In both there are systems of servers connected by the network (for simplicity, in cloud, we refer to a single datacenter as a ``server'', as we will not explore the parallelism inside a datacenter). The handling time on a server depends both on its performance metrics and its load. The final perceived latency comes from the network delays (required for transmitting the input data and the result) and from the handling time on the servers. Finally, in both cases every server has its initial load: in a CDN, the load is the current number of the data access requests to the server; in a cloud, the load is the number of initial tasks. We generalize these two problems to a load balancing of remote services.

We assume that in the balanced system, the handling time of a single request on a server linearly depends on the total number of requests to be processed by the server. A linear dependency reflects a constant throughput of a server. In real systems, increasing the level of concurrency too much may overload the server decreasing its throughput (trashing). However, assuming that the amount of work in the system as a whole is reasonable, there should be no overloaded servers in the balanced state. Similar assumptions are usually taken in congestion games~\cite{finiteCongestionGames, linearCongestionGames}
and in the queuing theory, where a linear dependency is expressed by Little's law.

We assume that the transmission duration of a single request does not depend on the number of sent requests. Although some models (e.g., routing games \cite{routingGames}) consider the cases when the bandwidth of a link may become a bottleneck, we focus on a widespread network, in which there are multiple cost-comparable routing paths between the servers. Thus, sending any data from one server to another should not significantly increase the network delay between them. These assumptions are also justified by our experiments -- in Appendix we discuss how the intensity of the network load, generated between the servers, influences the RTT between the servers in PlanetLab environment. In other words, we consider that the load our system imposes on the network is negligible: thus, the network delay is caused only by the latency (resulting from e.g., geographical distribution). Our problem formulation assumes the knowledge of such latencies; this is not a limitation because monitoring the pairwise latencies, which can change in time, is a well studied problem with known solutions (e.g. see \cite{Szymaniak04scalablecooperative, Chan-TinH11}). Optimizing latency is important for instance when streaming video files: a large latency delays the start, and, in case of communication problems, can be perceived as breaks in transmission.

Balancing of the servers loads and finding the mirroring minimizing network delays are analyzed in the literature, but usually separately (see Section~\ref{sec:related-work}). Distributed systems should consider both communication and computation.
On one hand, clouds get geographically distributed, thus cannot ignore network latency. On the other, a CDN handling complex, dynamically-created content of the modern web, can no longer ignore the load imposed by the requests. 

For delivering large static content, like multimedia, some currently used techniques cache the content at specially designated front-end servers. A particular server is chosen by the round-robin algorithm. This approach which is inefficient as, for instance, unpopular files are cached in multiple places. Benefits from optimizing requests redirections can be significant. For instance,~\cite{Chawla:2011:SCS:2002181.2002214} proposes a caching scheme that is both consistent (requests for the same content are redirected to the same front-end servers) and proportional (each server handles a desired proportion of requests). In this case, our algorithms can be viewed as a complementary optimization technique to caching -- once the content must be downloaded from the back-end servers, we show how to efficiently distribute the download requests.

In cloud computing, our model fits for instance processing streams of data in the real time or when the data stream is continuously produced and too large to be processed off-line. Consider a user interacting with a simulated virtual environment (e.g.~\cite{allard2007sofa}): user's actions are captured by cameras; their image streams are analyzed in the real time to build a 3D model; then this model interacts with the virtual world model. Other applications include extracting statistics on users' actions in the Internet; or image analysis.

In addition to a classic system with a central management, we analyze an organizationally-distributed system. Instead of a single, system-wide goal (minimize the overall request handling time), the organizations are selfishly interested only in optimizing the handling time of their local requests. This model reflects a CDN created as an agreement between e.g., many ISPs; or a federation of clouds, each having a different owner. Because typically the load changes dynamically, with peaks of demand followed by long periods of low activity, individual organizations are motivated to enter such a system: a peak can be offloaded, whereas handling foreign requests in the period of low activity is relatively inexpensive.

The lack of central coordination in the organiza\-tionally-distributed system increases the average processing time. The price of anarchy~\cite{worstCaseEqulibria} expresses the worst-case relative increase in the latency in comparison with relinquishing the control to a centrally-managed organization (like Akamai's CDN). As the price of anarchy varies considerably between systems (from relatively small in congestion games to unbounded in selfish replication~\cite{Rzadca2010ReplicaPlacementin}), we were curious to check it in our system.

Our contribution is the following:
\begin{inparaenum}[(i)] 
\item We show that the problem of network delay-aware load balancing can be stated as an optimization problem in the continuous domain; the problem is polynomially solvable, although standard solvers have $O(m^6)$ complexity\footnote{We do not claim that no better centralized algorithm exists; however, due to distributed nature of the problem, we are more interested in proposing a distributed algorithm, rather than tuning a centralized algorithm.}.
\item We propose a distributed algorithm that iteratively balances the servers' load towards the optimum. We confirm the algorithm's efficiency through simulation: even on a single CPU it outperforms the standard solvers.
\item In a network of selfish servers, we prove that the price of anarchy is low ($1 + O(2cs/l_{av})$) if the communication delay between each pair of servers is the same and the request handling time on a server is significantly higher then the network delay. The experiments show that the loss of performance caused by the selfishness remains low (below 1.15) also without these assumptions.
\end{inparaenum}

\vspace{-5pt}

\section{Mathematical model}\label{sec:mathematical-model}



\noindent
\textbf{Organizations, servers, tasks}\quad
The system consists of a set of $m$ \emph{organizations}, each owning a \emph{server} (or a \emph{cluster}) connected to the Internet. The servers are uniform; each server $i$ has a constant processing speed $s_i$. The $i$-th organization has its \emph{own load} consisting of a large number $n_{i}$ of small, individual \emph{tasks} (or \emph{requests}). 
The amount of load $n_i$ can be considered as a number of tasks at a particular time moment (snapshot); or, alternatively, as a steady state rate of incoming requests in a system continuously processing requests. 
A task corresponds to, e.g., in a computational cloud, a unit-size computation (e.g.: a single work unit in a BOINC-type application; or a single invocation of a map-reduce function); or, in a CDN, a request for remote data coming from a user assigned to server $i$ (typically, a user would be assigned to the closest server). 
In the basic model we assume that the small tasks have the same sizes (e.g. this corresponds to the divisible computation load; or in a CDN to the case where the stored data chunks have constant sizes); thus the execution of the single request on the $i$-th server takes $1/s_{i}$ time units. In Section \ref{sec:replica-placement} we show how to easily extend our results to the tasks of different sizes. \medskip

\noindent
\textbf{Relaying tasks, communication delays}\quad
Each organization can relay some of its own requests to other servers. If the request is relayed, the observed handling time is increased by the communication latency on the link.
We denote the communication latency between $i$-th and $j$-th server as $c_{ij}$ (with $c_{ii}=0$). 
Since communication delay of a single request does not depend on the amount of exchanged load (which is explained in Section \ref{sec::introduction} and which is confirmed by our experiments on PlanetLab -- see Appendix) $c_{ij}$ is just a constant instead of a function of the network load. We assume that the routing in the system is correct (optimized by the network layer). Thus, we will not consider optimizing communication time by relaying requests from $i$ to $j$ through a third server $k$ (if $c_{ik} + c_{kj} < c_{ij}$, the network layer would also discover the route $i \to k \to j$ and update the routing table accordingly, so that $c_{ij} := c_{ik} + c_{kj}$).
We assume that each request can be sent to and executed on any server. However, if we set some of the communication delays to infinity, we restrict the basic model to the case when each organization is allowed to relay its requests only to the given subset of the servers (its neighbors) which models e.g. the trust relationship. \medskip

\noindent
\textbf{Relay fractions, current loads}\quad
We use a fractional model in which a \emph{relay fraction} $\rho_{ij}$ denotes the fraction of the $i$-th organization's own requests that are sent (relayed) to be executed at $j$-th server ($\forall_{i,j} \: \rho_{ij} \geq 0$ and $\forall _{i} \: \sum_{j = 1}^{j = m}\rho_{ij} = 1$). The load balancing problem is to find the appropriate values of the relay fractions (formalized in the further part of this Section). Once the fractions are known, each organization knows how many of its own requests it should send to each server; the tasks are sent and executed at the appropriate servers. The fractional model might be considered as a relaxation of a problem of handling non-divisible requests; in Section~\ref{sec:replica-placement} we show how to round the solution of a fractional model to a discrete model. Moreover, the fractional model itself fits the divisible load model used in the scheduling theory.
For the sake of the clarity of the presentation we use the additional notation for the number of requests redirected from server $i$ to $j$ -- $r_{ij}$ (thus $r_{ij} = n_{i} \rho_{ij}$), and for the \emph{current load} of the server $i$, i.e. the number of requests relayed to $i$ by all other organizations, including the organization owning the server itself -- $l_i$ (thus, $l_i = \sum_{j = 1}^{m}r_{ji}$).
\medskip

\noindent
\textbf{Completion times}\quad
We don't assume any particular order of requests executed on a server. First, since the number of requests is large, considering any particular order on the servers would increase the computational complexity. Second, in a continuously running systems, we have no control over the order in which requests are produced (especially as they can be also delayed by the network); the usual FIFO policy results in an arbitrary order. 
Thus, for each of the $l_j$ request that are actually processed on $j$-th server, the expected processing time of each request is equal to $1/l_j \sum_1^{l_j} i/s_i = l_j/2s_{j}$ (constant omitted for clarity). Since $i$-th organization relayed $r_{ij}$ requests to $j$, the expected total completion time of requests relayed by $i$ to $j$ is equal to $r_{ij}(l_j/2s_{j} + c_{ij})$.
The expected total processing time $C_i$ of the $i$-th organization's own requests is a sum over all servers $j$ of the expected total completion times of the requests owned by $i$ and relayed to $j$.
\begin{equation}
C_{i} = \sum_{j = 1}^{m}\left(\frac{l_j}{2s_{j}} + c_{ij}\right)r_{ij} = \sum_{j = 1}^{m}\left(\left(c_{ij} + \sum_{k = 1}^{m}\frac{\rho_{kj} n_k}{2s_{i}}\right) \rho_{ij} n_{i}\right) \textrm{.}
\end{equation}

We consider the expected (or the average) processing time, rather than the makespan of an organization for several reasons. The average processing time is similar to the widely-used sum of processing times criterion ($\Sigma C_i$). We assume that the workload of each organization is created by many users. $\Sigma C_i$ models users' performance better than the makespan~\cite{dutot2004bi}. In all the contexts motivating our work from Section \ref{sec::introduction} (e.g., processing streams of data in the real time, delivering content to the users) we are focused on the average user performance. Also, while $C_{i}$ depends on the vector $\rho = [\rho_{kl}]$ quadratically, the relation between the makespan and $\rho$ is just linear, which makes the problem considerably easier. Thus, we believe that some of our results could be adapted for the cases when some different from the pointed applications of our model would require optimizing the makespan.

The total processing time of all the requests in the system is denoted as $\sum{C_{i}} = \sum_{i=1}^{i=m} C_{i}$.
\medskip

\noindent
\textbf{Problem formulation}\quad
We consider two related problems. First, the goal is to find such a vector of the fractions, $\rho$, that the total processing time of the requests, $\sum{C_{i}}$, is minimized. This goal corresponds to a centrally-managed system having a unique owner and a single goal.

Second, we analyze the case when servers are a common good, but each organization is selfishly minimizing the processing time of its own requests.
The $i$-th organization is responsible for sending its own requests to appropriate servers. In other words, the $i$-th organization adjusts the values of $\rho_{ij}$ in order to minimize  $C_{i}$.
This approach is similar to the selfish job model \cite{vocking-selfishloadbalancing}, in which jobs selfishly choose processors to minimize their execution time. Similar agreements exist in real-life systems: e.g., PlanetLab servers are treated as a common good managed by a central entity; PlanetLab users choose the servers they want to use for their experiments. Also, in academic grids (e.g. Grid5000 in France), participating organizations grant control over their resources to a central entity; in return, users can submit their jobs to any resource.
In this case,  we look for such a vector of the fractions $\rho$ for which the system reaches the Nash equilibrium. By comparing the resulting $\sum C_{i}$ with the result for the centrally-managed system, we will find the price of anarchy, quantifying the effect of selfishness on the total processing time.

\vspace{-5pt}
\section{Optimal solution}\label{sec:optimal-solution}

In this section we assume that there is a central processing unit that has the complete knowledge about the whole system. Given the communication latencies $c_{ij}$ and the organizations' own loads $n_{i}$, our goal is to find an algorithm setting relay fractions $\rho_{ij}$ so that the total processing time of all the requests $\sum{C_i}$ is minimized. We express the problem as a quadratic programming problem. We show that the problem is polynomially-solvable.

We express the total processing time $\sum{C_{i}}$ in a matrix form as 
$\sum{C_{i}} = \rho^{T} Q \rho + b^{T} \rho \text{,}$
where:
\begin{list}{$\bullet$}{\setlength{\leftmargin}{8pt} \setlength{\labelwidth}{0pt}
	\setlength{\itemsep}{0pt}\setlength{\parsep}{1pt}\setlength{\topsep}{1pt}\setlength{\partopsep}{1pt}}
\item $\rho$ is a vector of \emph{relay} fractions with $m \cdot m$ elements. $\rho_{(i,j)}$, the element at $(i \cdot m + j)$-th position, denotes the fraction of local requests of $i$-th server that are relayed to $j$-th server $\rho_{ij}$, thus: \\
    $\rho = [\rho_{(1,1)}, \rho_{(1,2)}, \ldots, \rho_{(1,m)}, \rho_{(2,1)}, \ldots, \rho_{(m,m)}]^{T} \textrm{;}$
\item $Q$ is $m^{2}$-by-$m^{2}$ matrix in which $q_{(i,j),(k,l)}$ denotes the element in $(i \cdot m + j)$-th row and in $(k \cdot m + l)$-th column:
  \begin{equation} \label{eq:q-def}
    q_{(i,j),(k,l)} = 
    \begin{cases}
      n_{i} n_{k}/s_{j} & \text{if $j = l$ and $i < k$;}\\
	  n_{i} n_{k}/2s_{j} & \text{if $j = l$ and $i = k$;}\\
      0 &  \text{otherwise;}\\
    \end{cases}
  \end{equation}
  Figure \ref{fig:matrixQ} presents the structure of matrix Q.
\item $b$ is a vector with $m^2$ elements with $b_{ij}$ denoting an element at $(i \cdot m + j)$-th position:
    $b_{(i,j)} = c_{ij} n_i \text{.}$

\end{list}



The following derivation shows how the matrix $Q$ is constructed:
\begin{align}
  \rho^{T} Q \rho & = \sum_{i,j}\rho_{(i,j)}\sum_{k \geq i}q_{(i,j), (k,j)}\rho_{(k,j)} \label{eq:q-der-1} \\
  & =
  \sum_{i,j}\rho_{(i,j)}(\sum_{k > i}\frac{n_{i}n_{k}\rho_{(k,j)}}{s_{j}} + \frac{n_{i}^{2}\rho_{(i,j)})}{2s_{j}} \label{eq:q-der-2} \\
  & =
  \sum_{i}\sum_{j}\sum_{k}\frac{n_{i}n_{k}\rho_{(i,j)}\rho_{(k,j)}}{2s_{j}} 
  \label{eq:q-der-3} =
  \sum_{i}\sum_{j}\frac{r_{ij}l_{j}}{2s_{j}} \textrm{.}
\end{align}
(\ref{eq:q-der-1}) follows from the construction of the matrix $Q$ (only elements $k \geq i$ are non-zero). (\ref{eq:q-der-2}) substitutes $q_{(i,j),(k,l)}$ with the values defined in (\ref{eq:q-def}). (\ref{eq:q-der-3}) uses commutativity of multiplication and substitutes $l_j = \sum_k n_k \rho_{(k,j)}$ and $r_{ij} = n_{i}\rho_{(i,j)}$.

\begin{figure}[t]
\setlength{\unitlength}{0.5mm}
\begin{picture}(75, 85) (0, 10)
\linethickness{0.05mm}
\put(10, 10){\line(0, 1){70}}
\put(90, 10){\line(0, 1){70}}

\put(10, 10){\line(1, 0){2}}
\put(10, 80){\line(1, 0){2}}
\put(90, 10){\line(-1, 0){2}}
\put(90, 80){\line(-1, 0){2}}

\put(13, 75){\tiny X}
\put(20, 68.7){\tiny X}
\put(27, 62.4){\tiny X}
\put(34, 56.1){\tiny X}
\put(41, 49.8){\tiny X}
\put(48, 43.5){\tiny X}
\put(55, 36.2){\tiny X}
\put(62, 29.9){\tiny X}
\put(69, 23.6){\tiny X}
\put(76, 18.3){\tiny X}
\put(83, 12){\tiny X}

\put(33, 65.6){\normalsize 0}
\put(47, 53){\normalsize 0}
\put(61, 40.7){\normalsize 0}
\put(75, 28.1){\normalsize 0}

\put(34, 75){\tiny X}
\put(41, 68.7){\tiny X}
\put(48, 62.4){\tiny X}
\put(55, 56.1){\tiny X}
\put(62, 49.8){\tiny X}
\put(69, 43.5){\tiny X}
\put(76, 36.2){\tiny X}
\put(83, 29.9){\tiny X}

\put(54, 65.6){\normalsize 0}
\put(68, 53){\normalsize 0}
\put(82, 40.7){\normalsize 0}

\put(55, 75){\tiny X}
\put(62, 68.7){\tiny X}
\put(69, 62.4){\tiny X}
\put(76, 56.1){\tiny X}
\put(83, 49.8){\tiny X}

\put(61, 76){\Large .}
\put(65, 76){\Large .}
\put(69, 76){\Large .}

\put(83, 57){\Large .}
\put(83, 61){\Large .}
\put(83, 65){\Large .}

\put(32, 25){\LARGE 0}

\put(13, 80){\vector(1, 0){20}}
\put(33, 80){\vector(-1, 0){20}}
\put(34, 80){\vector(1, 0){21}}
\put(55, 80){\vector(-1, 0){21}}

\put(92, 11){\vector(0, 1){18}}
\put(92, 29){\vector(0, -1){18}}
\put(92, 30){\vector(0, 1){19}}
\put(92, 49){\vector(0, -1){19}}

\put(21, 82){\scriptsize m}
\put(43, 82){\scriptsize m}
\put(94, 22){\scriptsize \begin{rotate}{270}m\end{rotate}}
\put(94, 41){\scriptsize \begin{rotate}{270}m\end{rotate}}

\linethickness{0.01mm}
\put(110, 49){\vector(-3, -1){31}}
\put(112, 49){\normalsize $n_{i}n_{k}/s_{j}$}
\put(110, 64){\vector(-3, -1){59}}
\put(112, 64){\normalsize $n_{i}n_{k}/2s_{j}$}

\end{picture}

\caption{Matrix Q: X denotes non-zero values}
\label{fig:matrixQ}
\vspace{-10pt}
\end{figure}

The constraints that $\rho_{ij}$ are the fractions ($\forall_{i,j} \: \rho_{ij} \geq 0$ and $\forall _{i} \: \sum_{j = 1}^{j = m}\rho_{ij} = 1$) can also be expressed in the matrix form. 
First, $\rho \geq 0_{m^2}$, where $0_{m^2}$ is a vector of length $m^{2}$ consisting of zeros.
Second, $A\rho = 1_{m}$, where $1_{m}$ is a vector of length $m$ and consisting of ones, and $A$ is a $m$-by-$m^2$ matrix defined by the following equation:
\begin{equation}
  a_{ij} = 
  \begin{cases}
    1 & \text{if  $i m \leq j < (i+1) m$} \\
    0 & \text{otherwise.}\\
  \end{cases}
\end{equation}

Minimization of $\sum{C_{i}}(\rho) = \rho^{T} Q \rho + b^{T} \rho$ with constraints $\rho \geq 0_{m^2}$ and $A\rho = 1_{m}$ is an instance of quadratic programing problem.
As an upper triangular matrix, matrix $Q$ has $m^{2}$ eigenvalues equal to the values at the diagonal: $n_{i}^{2}/2s_{j}$ ($1 \leq i,j \leq m$). All eigenvalues are positive so $Q$ is positive-definite. 
Thus,  the problem can be solved by the ellipsoid method in polynomial time \cite{quadraticSolvability}. According to \cite{convexComplexity}, the best running time reported for solving quadratic programing problem
with linear constraints is $O(n^3L)$ \cite{bestConvex}, where $L$ represents the total length of the input coefficients and $n$ the number of variables (here $n = m^2$), so the complexity of the best solution is $O(Lm^6)$.

\vspace{-7pt}
\section{Distributed algorithm}\label{sec:distr-algor}

\SetAlFnt{\footnotesize}
\SetAlCapFnt{\footnotesize}
\begin{algorithm}[t!]
  \SetKwInOut{Input}{input}
  \Input{$(i, j)$ -- the identifiers of the two servers}
  \KwData{$\forall_{k}$ $r_{ki}$ -- initialized to the number of requests owned by $k$ and relayed to $i$ ($\forall_{k}$ $r_{kj}$ is defined analogously)}
  \KwResult{The new values of $r_{ki}$ and $r_{kj}$}
  \ForEach{$k$}{
       $r_{ki} \leftarrow r_{ki} + r_{kj}$; $r_{kj} \leftarrow 0$\; 
  }
  $l_{i} \leftarrow \sum_{k} r_{ki}$ ; $l_{j} \leftarrow 0$ \;

  $servers$ $\leftarrow$ sort $[k]$ so that $c_{kj} - c_{ki} < c_{k'j} - c_{k'i}$ $\implies$ $k$ is before $k'$\;
  \ForEach{$k \in servers$}
  {
      $\Delta r_{ikj} \leftarrow \min\left(\frac{(s_{j}l_{i} - s_{i}l_{j}) - s_{i}s_{j}(c_{kj} - c_{ki})}{(s_{i} + s_{j})}, r_{ki}\right)$ \;
      \If{$\Delta r_{ikj} > 0$}
      {
         $r_{ki} \leftarrow r_{ki} - \Delta r_{ikj}$; $r_{kj} \leftarrow r_{kj} + \Delta r_{ikj}$ \;
         $l_{i} \leftarrow l_{i} - \Delta r_{ikj}$; $l_{j} \leftarrow l_{j} + \Delta r_{ikj}$ \;
      }
  }
  \Return{for each $k$: $r_{ki}$ and $r_{kj}$}
  \caption{\footnotesize \textsc{calcBestTransfer}(i, j)}
  \label{alg:exchangingLoads}
\end{algorithm}

\begin{algorithm}[t!]
  \SetKwFunction{calcBestTransfer}{calcBestTransfer}
  \SetKwFunction{transfer}{rely}
  \SetKwFunction{impr}{impr}
  \SetKwInOut{Notation}{Notation}
  \Notation{\impr(i, j) $\leftarrow$ this function calculates the improvement of $\sum C_{i}$ when transferring requests between $i$ and $j$. The number of requests that should be transferred can be computed by \calcBestTransfer{i, j}.}
  partner $\leftarrow \mathrm{argmax}_{j}(\impr(id, j))$\;
  \transfer(id, partner, \calcBestTransfer{id, partner})\;
  \caption{\footnotesize Min-Error (MinE) algorithm performed by server id.}
  \label{alg:distributedOptimal}
\end{algorithm}

The centralized algorithm requires the information about the whole network
 -- the size of the input data is $O(m^{2})$ and the $Q$ matrix has $O(m^{3})$ non-zero entries. A centralized algorithm has thus the following drawbacks:
(i) collecting information about the whole network is time-consuming; moreover, loads and latencies may frequently change;
(ii) a standard solver takes significant time (recall $O(Lm^6)$ in Section \ref{sec:optimal-solution});
(iii) the central algorithm is more vulnerable to failures. 
Motivated by these limitations we introduce a distributed algorithm for finding the optimal solution. 

The distributed algorithm requires that each server has up-to-date information about the loads on the other servers and about the communication delays from itself to the other servers (and not for all pairs of servers). Thus, for each server, the size of the input data is $O(m)$. As indicated in Section~\ref{sec::introduction}, the problem of monitoring the latencies is well-studied. The loads can be disseminated by a gossiping algorithm. As gossiping algorithms have logarithmic convergence time, if the gossiping is executed about $O(log(m))$ times more frequently than our algorithm, each server has accurate information about the loads. 

Each organization, $i$, keeps for each server, $k$, the information about the number of requests that were relied to $i$ by $k$. The algorithm iteratively improves the solution -- the $i$-th server in each step communicates with the locally optimal partner server -- $j$ (Algorithm~\ref{alg:distributedOptimal}). The pair $(i,j)$ locally optimizes the current solution by adjusting, for each $k$, $r_{ki}$ and $r_{kj}$ (Algorithm~\ref{alg:exchangingLoads}). 
In the first loop of the Algorithm~\ref{alg:exchangingLoads}, $i$, one of the servers, takes all the requests that were previously assigned to $i$ and to $j$. Next, all the organizations $[k]$ are sorted according to the ascending order of $(c_{kj} - c_{ki})$. The lower the value of $(c_{kj} - c_{ki})$, the more profitable it is to run requests of $k$ on $j$ rather than on $i$ in terms of the network topology. Then, for each $k$, the loads are balanced between servers $i$ and $j$.

In Section~\ref{sec:optimal-solution} we have shown that the optimization problem is convex. Thus, it is natural to try local optimization techniques. The presented mechanism requires only two servers involved in each optimization step, thus it is very robust to failures. This mechanism is similar in spirit to the diffusive load balancing~\cite{conf/ipps/AdolphsB12, Ackermann:2009:DAQ:1583991.1584046, Berenbrink:2011:DSL:2133036.2133152}; however there are substantial differences related to the fact that the machines are geographically distributed: (i) In each step no real requests are transferred between the servers; this process can be viewed as a simulation run to calculate the relay fractions $\rho_{ij}$. Once the fractions are calculated the requests are transfered and executed at the appropriate server. (ii) Each pair $(i, j)$ of servers exchanges not only its own requests but the requests of all servers that relayed their requests either to $i$ or to $j$. Since different servers may have different communication delays to $i$ and $j$ the local balancing requires more care (Algorithms~\ref{alg:exchangingLoads}~and ~\ref{alg:distributedOptimal}).

$\vspace{-5pt}$
\subsection{Correctness}

The following Lemma shows how to optimally exchange the requests owned by organization $k$ between a pair of servers $i$ and $j$.

\begin{lemma}\label{lemma::delegation}
  Consider two servers $i$ and $j$ that execute $r_{ki}$ and $r_{kj}$ requests of the $k$-th organization. The total processing time, $\sum{C_{i}}$, is minimized when the $k$-th server relies $\Delta r_{ikj}$ from $r_{ki}$ requests to be additionally executed on $j$-th server:
  \begin{align*}
    \Delta r'_{ikj} =  \frac{(s_{j}l_{i} - s_{i}l_{j}) - s_{i}s_{j}(c_{kj} - c_{ki})}{(s_{i} + s_{j})} \\
    \Delta r_{ikj} = \max (0, \min ( r_{ki}, \Delta r'_{ikj} ) )
  \end{align*}
\end{lemma}

\begin{proof}
If the $k$-th server moves some of its requests from $i$ to $j$, then it affects the completion time of all requests that were relayed either to $i$ or to $j$ (initial requests of all servers). Recall that $l_{i}$ and $l_{j}$ are the loads of the servers, respectively, $i$ and $j$, that is they include all tasks relayed to, respectively, $i$ and $j$. Thus, if $k$ removes $\Delta r$ of its requests from $i$, then the new processing time of all tasks on the server $i$ will be $(l_{i} - \Delta r)^{2}/2s_{i}$. Thus, we want to find $\Delta r_{ikj}$ that minimizes the function $f$:

  \begin{align*}
    f(\Delta r) = \frac{(l_{i} - \Delta r)^{2}}{2s_{i}} + \frac{(l_{j} + \Delta r)^{2}}{2s_{j}} - \Delta r c_{ki} + \Delta r c_{kj}
  \end{align*}

  We can find minimum by calculating derivative:

  \begin{align*}
    \frac{d f}{d \Delta r} = \frac{\Delta r - l_{i}}{s_{i}} + \frac{\Delta r + l_{j}}{s_{j}} - c_{ki} + c_{kj} = 0 \\
    \Delta r_{ikj} = \frac{(s_{j}l_{i} - s_{i}l_{j}) - s_{i}s_{j}(c_{kj} - c_{ki})}{(s_{i} + s_{j})}
  \end{align*}

  Also $\Delta r \in \langle 0, r_{ki} \rangle$, which proves the thesis.
\end{proof}

The following lemma proves the correctness of Algorithm~\ref{alg:exchangingLoads}.

\begin{lemma}\label{lemma::delegationOptimality}
  After execution of Algorithm~\ref{alg:exchangingLoads} for the pair of servers $i$ and $j$, it is not possible to improve $\sum C_{i}$ only by exchanging any requests between $i$ and $j$.
\end{lemma}
\begin{proof-sketch}
First we show that after the second loop no requests should be transferred from $i$ to $j$. For each organization $k$ the requests owned by $k$ were transferred from $i$ to $j$ in some iteration of the second loop; also, each of the next iterations of the second loop could only cause the increase of the load of $j$ (and decrease of $i$); thus transferring more requests of $k$ from $i$ to $j$ would be inefficient. Second, we will show that after the second loop no requests should be transferred back from $j$ to $i$ either. Let us take the last iteration of the second loop in which the requests of some organization $k$ were transferred from $i$ to $j$. After this transfer we know that $\Delta r_{ikj} = \frac{(s_{j}l_{i} - s_{i}l_{j}) - s_{i}s_{j}(c_{kj} - c_{ki})}{(s_{i} + s_{j})} \geq 0$ (otherwise the transfer would not be optimal). However, this implies that $\Delta r_{ik'j} = \frac{(s_{j}l_{i} - s_{i}l_{j}) - s_{i}s_{j}(c_{k'j} - c_{k'i})}{(s_{i} + s_{j})} \geq 0$ for each server $k'$ considered before $k$. As $\Delta r_{ik'j} \geq 0$ we get $\Delta r_{jk'i} \leq 0$.
\end{proof-sketch}

\subsection{Error estimation}\label{sec:convergence}
The following analysis bounds the distance of the current solution of the distributed algorithm to the optimum as a function of the disparity of servers' load. When running the algorithm, this result can be used to assess whether it is still profitable to continue running the algorithm: if the load disparity is low, the current solution is close to the optimum.

We introduce the following notation for the analysis. 
$\rho'$ is the snapshot (the current solution) derived by distributed algorithm. 
$\rho$ is the optimal solution that minimizes $\sum C_{i}$ (if there are multiple optimal solutions with the same $\sum{C_i}$ , $\rho$ is the closest solution to $\rho'$ in the Manhattan metric).
$(P, \Delta \rho)$ is a weighted, directed \emph{error graph}: $\Delta \rho[i][j]$
indicates the number of requests that should be transferred from server $i$ to $j$ in order to reach $\rho$ from $\rho'$ ($\Delta \rho[i][j]$ requests either belong to $i$, or to $j$, and not to another server $k$). 
We define $\dir$ as the \emph{direction of transport}: $\dir(i, j)=1$ if $i$ transfers to $j$ its own requests; $\dir(i, j)=-1$ if $i$ returns to $j$ the requests that initially belonged to $j$. Let $succ(i)$ denotes the set of successors in the error graph: $succ(i) = \{j: \Delta \rho[i][j] > 0\}$; $prec(i)$ denotes the set of predecessors: $prec(i) = \{j: \Delta \rho[j][i] > 0\}$.

In the error graph, a \emph{negative cycle} is a sequence of servers $i_{1}, i_{2}, \ldots, i_{n}$ such that 
\begin{inparaenum}[(i)]
\item $i_{1} = i_{n};$
\item $\forall_{j \in \{1,\ldots n-1\}}$ $\Delta \rho[i_{j}][i_{j+1}] > 0$; and
\item $\sum_{j=1}^{n-1}\dir(i_{j}, i_{j+1}) c_{i_{j}i_{j+1}} < 0$.
\end{inparaenum} 

A negative cycle is sequence of servers that essentially redirect their requests to one another. A solution without negative cycles has a smaller processing time: after dismantling a negative cycle, loads on servers remain the same, but the communication time is reduced. In Appendix, we show how to detect and remove negative cycles; in order to simplify the presentation of the subsequent analysis, we consider that there are no negative cycles.

\begin{proposition}\label{lemma::convergence}
  If \begin{inparaenum}[(i)]
  \item the error graph $\Delta \rho$ has no negative cycle; and
  \item $\sum_{j} \max_{k} ((\frac{1}{s_j} + \frac{1}{s_k})\Delta r_{jk}) = \Delta R$
    ($\Delta r_{ij}$ is the number of requests which in the current state $\rho'$ would be relied to $j$-th server by the $i$-th server (as the result of Algorithm~\ref{alg:exchangingLoads}),
  \end{inparaenum}
  then $\| \rho - \rho' \|_{1} \leq (4m + 1)\Delta R\sum_i s_{i}$, where $\| \cdot \|_{1}$ denotes the Manhattan metric.
\end{proposition}

\begin{proof}
  First we show that there is no cycle in the error graph. By contradiction let us assume that there is a cycle: $i_1, \ldots, i_{n-1}, i_{n}$ (with $i_1=i_n$).
  Because the error graph has no negative cycle, we have: $\sum_{j=1}^{n-1} \dir(i_{j}, i_{j+1}) c_{i_{j}, i_{j+1}} \geq 0$. 
  Now, if we reduce the number of requests sent on each edge of the cycle:
  \begin{align*}
    \Delta \rho [i_{j}, i_{j+1}] :=
    \Delta \rho [i_{j}, i_{j+1}] - min_{k \in \{1,\ldots, n-1\}}( \rho [i_{k}, i_{k+1}])
  \end{align*}
  then the load of the servers $i_{j}, j \in \{1, \ldots, n-1 \}$ will not change (each server both receives and sends $min_{k \in \{1,\ldots, n-1\}}( \rho [i_{k}, i_{k+1}])$ requests less).
  Additionally, the transfers in the network are reduced. Thus, we get a new optimal solution which is closer to $\rho'$ in Manhattan metric, which contradicts that $\rho$ is the optimal.

  Second, we show how to bound the distance to the optimum solution on a server by transfers and loads on neighbors.
  Let $l_{i}$ be the load of the $i$-th server in the optimal solution $\rho$. Let $l'_{i}$ be the load of the $i$-th server in state $\rho'$. 
  Consider a server $i$ for which $l_{i} \leq l'_{i}$. 
  In state $\rho'$, in order to balance $i$ and $j \in succ(i)$, at most $\Delta r_{ij}$ requests must be transferred (Lemma~\ref{lemma::delegationOptimality}).
  For each $k$, $\Delta r_{ikj}$ depends on the difference between weighted loads $l'_is_{j} - l'_js_{i}$ (see the thesis of Lemma~\ref{lemma::delegation}). 
Thus, by Lemma~\ref{lemma::delegationOptimality}, in the current state, $i$ and $j$ would be balanced if the difference in weighted loads is at most $D = (l'_i - \Delta r_{ij})s_{j} - (l'_j + \Delta r_{ij})s_{i}$. 
In the optimal state, the weighted load of $j$ is at least $s_{j}(l'_j - \max((l'_{j} - l_{j}) , 0))$.
In the optimal state, $j$ must be also balanced with $i$, thus, the difference in weighted loads is at most $D$. 
By solving for the reduction of load on $i$, we get that $l'_i$ is decreased by at most $\frac{s_{i} + s_{j}}{s_{j}} \Delta r_{ij} + \max(\frac{s_{i}}{s_{j}}(l'_{j} - l_{j}), 0)$.
  In the optimal solution, all the pairs of servers are balanced; thus, the difference between the current and the optimal load can be bounded by:
  \begin{equation}\label{equation::transferEstimation}
  \begin{aligned}
    \frac{1}{s_{i}}(l'_{i} - l_{i}) \leq \max_{j \in succ(i)}((\frac{1}{s_{i}} + \frac{1}{s_{j}})\Delta r_{ij} +
													 \frac{1}{s_{j}}\max(l'_{j} - l_{j}, 0))
  \end{aligned}
  \end{equation}

  Eq.~\ref{equation::transferEstimation} holds also for any server $j \in succ(i)$ for which  $l_{j} \leq l'_{j}$. It can be recursively expanded until reaching the servers without  successors in the error graph (there are no cycles in the graph). The resulting expanded equation takes into account the path constructed by the maximum imbalance ($\max_{j \in succ(i)}$); the cost of this path is bounded by the cost of all the imbalances, which leads to:
  \begin{equation*}
    l'_{i} - l_{i} \leq s_{i}\sum_{j} \max_{k} \left((\frac{1}{s_j} + \frac{1}{s_k})\Delta r_{jk}\right) = s_{i}\Delta R
  \end{equation*}

  For the servers with $l_{i} > l'_{i}$, we can obtain the similar estimation by taking the set of predecessors $prec$ and expanding the inequalities towards the servers that have no predecessors instead of moving towards those without successors.  

  Eq.~\ref{equation::transferEstimation} considered \emph{loads} $l_i$; the imbalance of transfer $\Delta \rho[i][j]$ can be similarly bounded by $\Delta \rho[i][j] \leq (\Delta r_{ij} + \frac{s_{j}}{s_{i} + s_{j}} (|l'_{i} - l_{i}| + |l'_j - l_j|)$): $ \Delta r_{ij}$ is what would be transfered in the current state; and $|l'_{i} - l_{i}|$ takes into account the transfers that might be triggered by the future changes in the load; the transfers are proportional to the relative speed $\frac{s_{j}}{s_{i} + s_{j}}$. 
  Finally:
  \begin{align*}
    \| \rho - \rho' \|_{1} & = \sum_{i} \sum_{j}\Delta \rho[i][j] \leq \sum_{i} \sum_{j} (\Delta r_{ij} + 4s_{i}\Delta R) \\
    & \leq (4m + 1)\Delta R\sum_i s_{i}
  \end{align*}
\end{proof}

Proposition~\ref{lemma::convergence} gives the estimation of the error for such partial solutions that do not have a negative cycles. Therefore the algorithm that cancels negative cycles (see Appendix) should be run whenever the estimation for distance to the optimal solution is needed. Our experiments show, however, that the negative cycles are rare in practice and that pure Algorithm~\ref{alg:distributedOptimal} can remove them efficiently (Section~\ref{sec:experiments}).

\section{Selfish organizations}\label{sec:selfish-servers}

In this section we consider the case when the organizations are acting selfishly -- the $i$-th of them tries to minimize the total processing time of its own requests -- $C_{i}$.  We are interested in a steady state in which all the peers have no interest in redirecting any of its requests to different servers -- the Nash equilibrium. 

\subsection{Homogeneous network}\label{sec:homogenous-network}

In this section we present the characteristic of the Nash equilibrium in case when all the servers have equal processing power ($\forall_{i} \,\, s_{i} = s$), and when all the connections between servers have the same communication delay ($\forall_{ij} \,\, c_{ij} = c$). We consider homogeneous model, as the modeling of a heterogeneous interconnection graph is complex. The simulation experiments (Section~\ref{sec:experimentalPriceOfAnarchy}) show that in the case of selfish servers the average relative degradation of the system goal on heterogeneous networks is similar to, or lower than on the homogeneous networks.

\begin{lemma}\label{lemma::imbalance}
  For every two servers $i$ and $j$ the difference between their average loads is bounded: $|l_{i} - l_{j}| \leq c \cdot s$  
\end{lemma}

\begin{proof}
  (by contradiction) Assume $|l_{i} - l_{j}| > c\cdot s$. Without loosing the generality, $l_{i} > l_{j}$. Recall that $r_{ij}$ is the number of redirected requests $r_{ij} = n_{i} \rho_{ij}$. For each sever $k$ ($k \neq i$), it is not profitable to put more of its requests to the more loaded server, so $r_{kj} \geq r_{ki}$. Now we want to find the relation between $l_{i}, l_{j}, r_{ij}$ and $r_{ii}$. In a Nash equilibrium,  it is not profitable for $i$ to redirect any additional $x$ of its own requests from itself to $j$, which can be formally expressed by the equation:
  \begin{align*}
    0 & \leq \: \frac{(l_{i} - x)(r_{ii} - x)}{2s} + \frac{(l_{j} + x )(r_{ij} + x)}{2s} + c(r_{ij} + x) \\
      & - \frac{l_{i}r_{ii}}{2s} - \frac{l_{j}r_{ij}}{2s} - cr_{ij} \text{,}
  \end{align*}
  equivalent to:
  \begin{align*}
    r_{ij} - r_{ii} + 2x & \geq l_{i} - l_{j} - 2c \cdot s \text{.} \\
    \intertext{Because the inequality must hold for every positive $x$, and because $l_{i} - l_{j} > c \cdot s$}
    r_{ij} - r_{ii} & > c \cdot s - 2c \cdot s = -c \cdot s
    \intertext{Now we can show the contradiction, because}
    l_{j} = \sum_{k = 1}^{k = m}r_{kj} & > \sum_{k = 1}^{k = m}r_{ki} - c \cdot s = l_{i} - c \cdot s
    \intertext{from which it follows that}
    l_{i} - l_{j} & < c \cdot s \text{.}
  \end{align*}
\end{proof}

Let us denote the average load on the server as $l_{av}$, thus $l_{av} = \frac{1}{m}\sum_{i = 1}^{i = m}l_{i}$. The following theorem gives the tight estimation of the price of anarchy when the servers are loaded compared to the delay ($l_{av} \gg 2cs$). (If the servers are not loaded, our estimation of the price of anarchy is dominated by $O((\frac{c}{l_{av}})^{2})$ element).

\begin{theorem}
  The price of anarchy in the homogeneous network is:
  $PoA = 1 + \frac{2cs}{l_{av}} + O((\frac{cs}{l_{av}})^{2})$.
\end{theorem}

\begin{proof}
  (upper bound) We denote the load imbalance on the $i$-th server as $\Delta_{i} = l_{i} - l_{av}$.  It follows that $\sum_{i = 1}^{i = m}\Delta_{i} = 0$. 
  Also, from Lemma~\ref{lemma::imbalance} we have $\Delta_{i} \leq c \cdot s$. Additionally, each request can be relied at most once, thus the total time used for communication is bounded by $m l_{av} c$.
  Therefore, the total processing time in case of selfish peers, $\sum C_i(\text{self})$ is bounded:
  \begin{align*}
    \sum C_i (\text{self}) & \leq ml_{av}c + \sum_{i }\frac{(l_{i})^2}{2s}
                           = ml_{av}c + \sum_{i}\frac{(l_{av} + \Delta_{i})^2}{2s} \\
                           & = \frac{ml_{av}^2}{2s} + \sum_{i}\frac{\Delta_{i}^{2}}{2s} + ml_{av}c
                           \leq \frac{ml_{av}^2}{2s} + \frac{mc^2s}{2} + ml_{av}c
  \end{align*}

  The total processing time is the smallest when the servers have equal load (each server processes exactly $l_{av}$ requests) and do not communicate, thus the optimum is bounded by $(\sum C_i)^* \geq \frac{m l_{av}^2}{2s}$.

  Thus, the price of anarchy is bounded by:
  \begin{equation*}
    PoA \leq \frac{ml_{av}^2 +  2ml_{av}cs + mc^2s^2}{ ml_{av}^2} = 1 + \frac{2cs}{l_{av}} + (\frac{cs}{l_{av}})^2
  \end{equation*}

  (tightness) Consider an instance with servers having equal initial load: $\forall_{i}\,\, n_{i} = l_{av}$. 

  In the optimal solution no requests will be redirected. 

  When servers are selfish, the $i$-th server will redirect to $j$-th server ($i \neq j$) $\frac{l_{av} - 2c \cdot s}{m}$ requests and will execute $(2c \cdot s + \frac{l_{av} - 2c \cdot s}{m}$ of its own requests on itself. As a result: $l_i = l_{av}$.

  This is a Nash Equilibrium state, because it is not profitable for any server to redirect any $x$ more of its own requests to the other server, nor to execute any $x$ more requests on itself instead of some other server, as the two following inequalities hold for every positive $x$:
  \begin{align*}
    0 & < \frac{l_{av} - x}{2s}(2c \cdot s + \frac{l_{av} - 2c \cdot s}{m} - x)
       + \frac{l_{av} + x}{2s}(\frac{l_{av} - 2c \cdot s}{m} + x) \\
       & + cx -\frac{l_{av}}{2s}(2c \cdot s + \frac{l_{av} - 2c \cdot s}{m}) -\frac{l_{av}}{2s}(\frac{l_{av} - 2c \cdot s}{m}) \\
    0 & < \frac{l_{av} + x}{2s}(2c \cdot s + \frac{l_{av} - 2c \cdot s}{m} + x)
       + \frac{l_{av} - x}{2s}(\frac{l_{av} - 2c \cdot s}{m} - x) \\
       & - cx -\frac{l_{av}}{2s}(2c \cdot s + \frac{l_{av} - 2c \cdot s}{m}) -\frac{l_{av}}{2s}(\frac{l_{av} - 2c \cdot s}{m})
  \end{align*}

  Thus, we get the lower bound on the price of anarchy:
  \begin{align*}
    PoA & \geq \frac{ml_{av}^{2} + m(l_{av} - 2c \cdot s - \frac{l_{av} - 2c\cdot s}{m})c \cdot 2s}{ml_{av}^{2}} \notag \\
       & = 1 + \frac{2cs}{l_{av}} - 4(\frac{sc}{l_{av}})^{2} - \frac{2(l_{av} - 2c^2s^2)}{ml_{av}^{2}} \notag
        \geq 1 + \frac{2cs}{l_{av}} - 4(\frac{cs}{l_{av}})^{2} \textrm{.}
  \end{align*}
  Summarizing:
  \begin{align*}
  1 + \frac{2cs}{l_{av}} - 4(\frac{cs}{l_{av}})^{2} \leq PoA \leq 1 + \frac{2cs}{l_{av}} + (\frac{cs}{l_{av}})^2
  \end{align*}
\end{proof}

The price of anarchy depends on the average load on the server and on the network delay. For the more general case, in Section \ref{sec:experimentalPriceOfAnarchy} we present the estimations derived from simulations.

\vspace{-7pt}
\section{Simulation Experiments}\label{sec:experiments}

In this section we show the results from the two groups of experiments. First, we investigate convergence time of the distributed algorithm. Second, we assess the loss of performance in an organizationally-distributed system compared to the optimal, central solution. The loss is computed as a ratio of the total processing times. 

\vspace{-5pt}
\subsection{Settings}\label{sec:settings}

We experimented on two kinds of networks: homogeneous, with equal communication latencies ($c_{ij}$ = 20); and heterogeneous, where latencies were based on measurements between PlanetLab nodes\footnote{\url{http://iplane.cs.washington.edu/data/data.html}} expressed in milliseconds\footnote{The dataset does not contain latencies for all pairs of nodes, so we had to complement the data by calculating minimal distances.}.

In the initial experiments, we analyzed networks composed of 20, 30, 50, 100, 200 and 300 serves. We also performed some experiments on larger networks (500, 1000, 2000, 3000 servers). The processing speeds of the servers $s_{i}$ were uniformly distributed on the interval $\langle 1, 5 \rangle$.

We conducted the experiments for exponential and uniform distribution of the initial load over the servers. For each distribution we analyzed five cases with the average load equal to 10, 20, 50, 200 and 1000 requests (assuming that processing a single request on a single server takes 1ms). We also analyzed the case of peak distribution -- with 100.000 requests owned by a single server.


We evaluated the result based on the distance to the optimal solution, which because of the $O(m^6)$ complexity of standard solvers (see Section \ref{sec:optimal-solution}) was approximated by our distributed algorithm.


\begin{figure}[t!]
  \centering
  \includegraphics[width=0.47\textwidth]{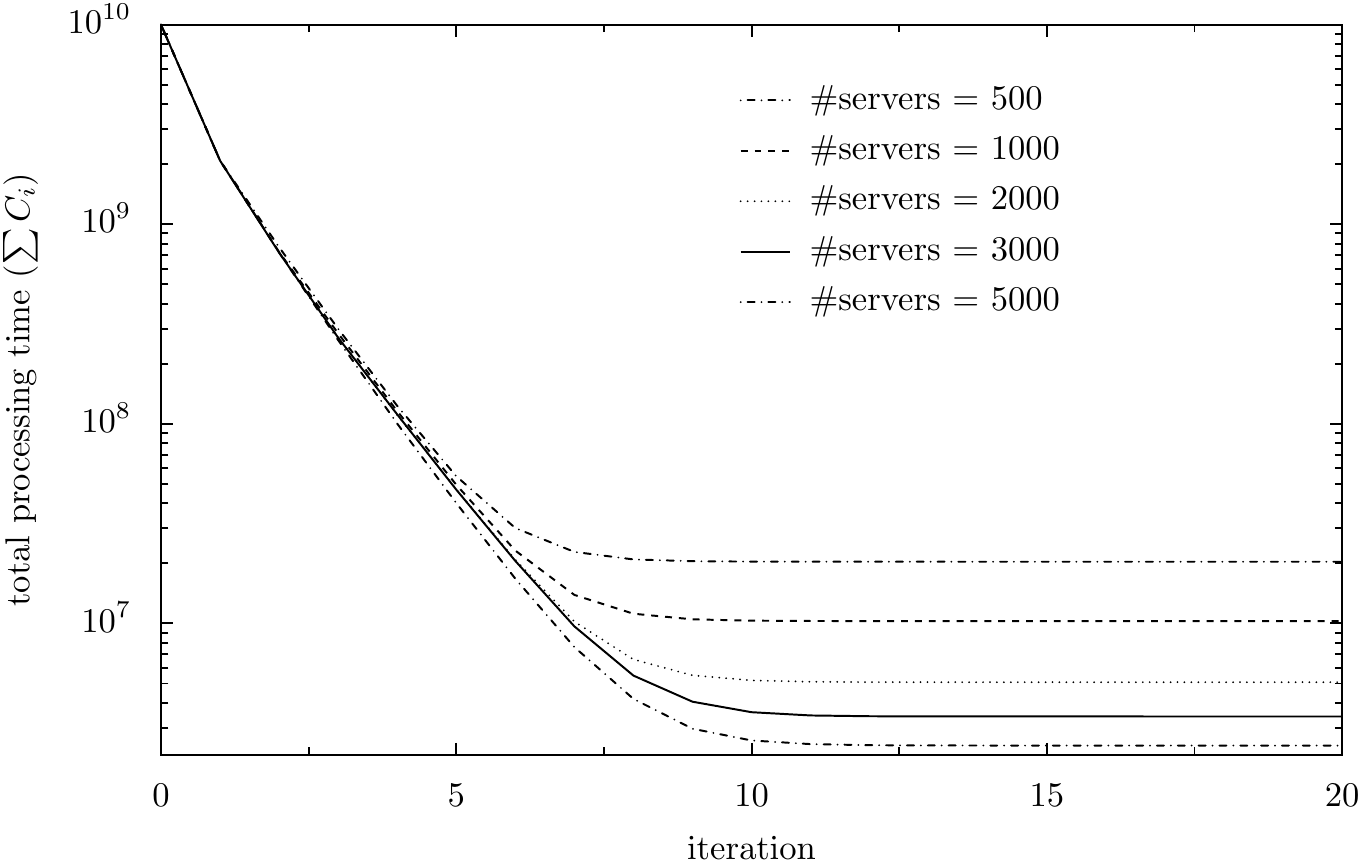}
  \caption{The convergence of the distributed algorithm for peak distribution of initial loads.}
  \label{fig:convergence}
  \vspace{-15pt}
\end{figure}

\vspace{-7pt}
\subsection{Convergence time of the distributed algorithm}\label{sec:exp-distr-algor}

In the first series of experiments, we evaluated the efficiency of the distributed algorithm measured as the number of iterations the algorithm must perform in order to decrease the difference between the total processing times in the current and the optimal requests distributions to less than 2\% of the average load. In a single iteration of the distributed algorithm, each server executes Algorithm~\ref{alg:distributedOptimal}; if there were many pairs of the servers to be optimized we run optimization in the random order. Table~\ref{table:algorithmConvergence} summarizes the results.

The results indicate that the number of iterations mostly depends on the size of the network and on the distribution of the initial load. The type of the network (planet-lab vs. homogeneous) does not influence the convergence time. Larger networks and peak distribution result in higher convergence times. In all considered networks, the algorithm converged in at most 9 iterations. 

Next, we decreased the required precision error from 2\% to 0.1\%, and ran the same experiments. The results are given in Table \ref{table:algorithmConvergence2}. In this case, similarly, the required number of iterations was the highest for peak distribution of the initial load. In each case the algorithm converged in at most 11 iterations. Even for 300 servers the average number of iterations is below 8. Also, the standard deviations are low, which indicates that the algorithm is stable with respect to its fast convergence.

\begin{table}[t!]
  \small
  \centering
  \begin{tabular}{cc|c|c|c|}
    \cline{3-5}
    & & \multicolumn{3}{|c|}{\# iterations} \\
    \cline{3-5}
    & & average & max & st. dev. \\
    \cline{1-5}
    \multicolumn{1}{|c|}{\multirow{2}{*}{$m$ $\leq$ 50}} &
    \multicolumn{1}{|c|}{uniform} & 1.65 & 3 & 0.49     \\
    \cline{2-5}
    \multicolumn{1}{|c|}{} &
    \multicolumn{1}{|c|}{exp.} & 2.35 & 3 & 0.47   \\
    \cline{2-5}
    \multicolumn{1}{|c|}{} &
    \multicolumn{1}{|c|}{peak} & 4.87 & 6 & 0.71    \\
    \cline{1-5}
    \multicolumn{1}{|c|}{\multirow{2}{*}{$m$ = 100}} &
    \multicolumn{1}{|c|}{uniform} & 2.0 & 2.0 & 0.0     \\
    \cline{2-5}
    \multicolumn{1}{|c|}{} &
    \multicolumn{1}{|c|}{exp.} & 2.62 & 3 & 0.48  \\
    \cline{2-5}
    \multicolumn{1}{|c|}{} &
    \multicolumn{1}{|c|}{peak} & 6.88 & 7 & 0.32    \\
    \cline{1-5}
    \multicolumn{1}{|c|}{\multirow{2}{*}{$m$ = 200}} &
    \multicolumn{1}{|c|}{uniform} & 2.1 & 3 & 0.33     \\
    \cline{2-5}
    \multicolumn{1}{|c|}{} &
    \multicolumn{1}{|c|}{exp.} & 3.1 & 4 & 0.33   \\
    \cline{2-5}
    \multicolumn{1}{|c|}{} &
    \multicolumn{1}{|c|}{peak} & 7.84 & 8 & 0.37    \\
    \cline{1-5}
    \multicolumn{1}{|c|}{\multirow{2}{*}{$m$ = 300}} &
    \multicolumn{1}{|c|}{uniform} & 2.0 & 2 & 0.0     \\
    \cline{2-5}
    \multicolumn{1}{|c|}{} &
    \multicolumn{1}{|c|}{exp.} & 3.25 & 4 & 0.43   \\
    \cline{2-5}
    \multicolumn{1}{|c|}{} &
    \multicolumn{1}{|c|}{peak} &  8.0 & 8 & 0.0   \\
    \cline{1-5}
  \end{tabular}
  \caption{The number of iterations of the distributed algorithm required to obtain at most 2\% relative error in the total processing time $\Sigma C_i$.}
  \label{table:algorithmConvergence}
\end{table}
\begin{table}[t]
  \small
  \centering
  \begin{tabular}{cc|c|c|c|}
    \cline{3-5}
    & & \multicolumn{3}{|c|}{\# iterations} \\
    \cline{3-5}
    & & average & max & st. dev. \\
    \cline{1-5}
    \multicolumn{1}{|c|}{\multirow{2}{*}{$m$ $\leq$ 50}} &
    \multicolumn{1}{|c|}{uniform} & 5.1 & 7 & 1.0     \\
    \cline{2-5}
    \multicolumn{1}{|c|}{} &
    \multicolumn{1}{|c|}{exp} & 5.5 & 7 & 0.9   \\
    \cline{2-5}
    \multicolumn{1}{|c|}{} &
    \multicolumn{1}{|c|}{peak} & 6.4 & 7 & 0.5    \\
    \cline{1-5}
    \multicolumn{1}{|c|}{\multirow{2}{*}{$m$ = 100}} &
    \multicolumn{1}{|c|}{uniform} & 5.8 & 9 & 1.6     \\
    \cline{2-5}
    \multicolumn{1}{|c|}{} &
    \multicolumn{1}{|c|}{exp.} & 6.3 & 9 & 1.5  \\
    \cline{2-5}
    \multicolumn{1}{|c|}{} &
    \multicolumn{1}{|c|}{peak} & 8.0 & 9 & 0.2    \\
    \cline{1-5}
    \multicolumn{1}{|c|}{\multirow{2}{*}{$m$ = 200}} &
    \multicolumn{1}{|c|}{uniform} & 6.1 & 9 & 2.2     \\
    \cline{2-5}
    \multicolumn{1}{|c|}{} &
    \multicolumn{1}{|c|}{exp.} & 7.1 & 10 & 2.0   \\
    \cline{2-5}
    \multicolumn{1}{|c|}{} &
    \multicolumn{1}{|c|}{peak} & 9.9 & 10 & 0.3    \\
    \cline{1-5}
    \multicolumn{1}{|c|}{\multirow{2}{*}{$m$ = 300}} &
    \multicolumn{1}{|c|}{uniform} & 6.2 & 10 & 2.4     \\
    \cline{2-5}
    \multicolumn{1}{|c|}{} &
    \multicolumn{1}{|c|}{exp.} & 7.7 & 11 & 2.0   \\
    \cline{2-5}
    \multicolumn{1}{|c|}{} &
    \multicolumn{1}{|c|}{peak} &  10.0 & 10 & 0.0   \\
    \cline{1-5}
  \end{tabular}
  \caption{The number of iterations of the distributed algorithm required to obtain at most 0.1\% relative error in the total processing time $\Sigma C_i$.}
  \label{table:algorithmConvergence2}
  \vspace{-10pt}
\end{table}


Also, we assessed whether a variation of the distributed algorithm that does not eliminate negative cycles (Appendix~\ref{sec:remov-negat-cycl}) has a slower convergence time. Although required to prove the convergence (Section~\ref{sec:convergence}), eliminating the negative cycles is complex in implementation and dominates the execution time.

We compared two versions of the distributed algorithm: without negative cycle removal; and with the removal every two iterations of the algorithm. 
The number of iterations for two versions of the algorithm were \emph{exactly} the same in all 6000 experiments. These result show that the cycles which happen in practice can be efficiently removed by pure Algorithm~\ref{alg:exchangingLoads}. Also, the negative cycles are rare in practice.

Finally, we analyzed the convergence of the distributed algorithm without negative cycles elimination on larger networks (Figure~\ref{fig:convergence}). The previous experiments shown that the algorithm convergence is the slowest for peak distribution of the initial load, therefore we chose this case for the analysis. The experiments used heterogeneous network. 
The results indicate that even for larger networks the total processing time decreases exponentially.

\begin{table}[t!]
  \centering
  \small
  \begin{tabular}{ccc|c|c|c|}
    \cline{4-6}
    & & & \multicolumn{3}{|c|}{Ratio} \\
    \cline{4-6}
    & & & avg. & max & st. dev. \\
    \cline{1-6}
    \multicolumn{1}{|c|}{\multirow{6}{*}{\begin{sideways}const $s_{i}$\end{sideways}}} &
    \multicolumn{1}{|c|}{\multirow{2}{*}{$l_{av}$ $\leq$ 30}} &
    \multicolumn{1}{|c|}{$c_{ij}$ = 20} & 1.041 & 1.098 & 0.029     \\
    \cline{3-6}
    \multicolumn{1}{|c|}{} & \multicolumn{1}{|c|}{} &
    \multicolumn{1}{|c|}{PL} & 1.014 & 1.049 & 0.007     \\
    \cline{2-6}
    \multicolumn{1}{|c|}{} & \multicolumn{1}{|c|}{\multirow{2}{*}{$l_{av}$ = 50}} &
    \multicolumn{1}{|c|}{$c_{ij}$ = 20} & 1.114 & 1.150 & 0.031  \\
    \cline{3-6}
    \multicolumn{1}{|c|}{} & \multicolumn{1}{|c|}{} &
    \multicolumn{1}{|c|}{PL} & 1.011 & 1.033 & 0.006 \\
    \cline{2-6}
    \multicolumn{1}{|c|}{} & \multicolumn{1}{|c|}{\multirow{2}{*}{$l_{av}$ $\geq$ 200}} &
    \multicolumn{1}{|c|}{$c_{ij}$ = 20} & 1.024 & 1.055 & 0.018 \\
    \cline{3-6}
    \multicolumn{1}{|c|}{} & \multicolumn{1}{|c|}{} &
    \multicolumn{1}{|c|}{PL} & 1.003 & 1.022 & 0.003 \\
    \cline{1-6}
    \multicolumn{1}{|c|}{\multirow{6}{*}{\begin{sideways}uniform $s_{i}$\end{sideways}}} &
    \multicolumn{1}{|c|}{\multirow{2}{*}{$l_{av}$ $\leq$ 30}} &
    \multicolumn{1}{|c|}{$c_{ij}$ = 20} & 1.000 & 1.022 & 0.001     \\
    \cline{3-6}
    \multicolumn{1}{|c|}{} & \multicolumn{1}{|c|}{} &
    \multicolumn{1}{|c|}{PL} & 1.000 & 1.000 & 0.000     \\
    \cline{2-6}
    \multicolumn{1}{|c|}{} & \multicolumn{1}{|c|}{\multirow{2}{*}{$l_{av}$ = 50}} &
    \multicolumn{1}{|c|}{$c_{ij}$ = 20} & 1.041 & 1.062 & 0.018  \\
    \cline{3-6}
    \multicolumn{1}{|c|}{} & \multicolumn{1}{|c|}{} &
    \multicolumn{1}{|c|}{PL} & 1.000 & 1.000 & 0.000 \\
    \cline{2-6}
    \multicolumn{1}{|c|}{} & \multicolumn{1}{|c|}{\multirow{2}{*}{$l_{av}$ $\geq$ 200}} &
    \multicolumn{1}{|c|}{$c_{ij}$ = 20} & 1.001 & 1.029 & 0.006 \\
    \cline{3-6}
    \multicolumn{1}{|c|}{} & \multicolumn{1}{|c|}{} &
    \multicolumn{1}{|c|}{PL} & 1.000 & 1.000 & 0.000 \\
    \cline{1-6}
  \end{tabular}
  \caption{Experimental assessment of the cost of selfishness: ratios between total processing times in cases of selfish and cooperative servers.}
  \label{table:priceOfAnarchy}
  \vspace{-7pt}
\end{table}


\vspace{-5pt}
\subsection{Cost of selfishness}\label{sec:experimentalPriceOfAnarchy}

In the second series of experiments we experimentally measured the cost of selfishness as the ratio between total processing times in cases of selfish and cooperative servers (Table~\ref{table:priceOfAnarchy}). In each experiment, the Nash equilibrium was approximated by the following heuristics. Each server was playing its best response to the current distribution of requests. We terminated when all servers in two consecutive steps changed the distribution of their requests by less than 1\%.
We computed the ratio of the total processing times: the (approximated) Nash equilibrium to the optimal value.

The cost of selfishness is low. The average is below $1.06$; and the maximal value is below $1.15$. The estimation of the cost of selfishness is higher in case of constant processing rates $s_{i}$. It additionally depends on the ratio between the average initial load and the network latency and on the structure of the network. The highest cost is for homogeneous networks with constant processing rates and having medium initial load about 2 times longer than the mean communication delay. The experiments show that the cost of selfishness is independent of the size of the network and the type of distribution of initial loads.  

\vspace{-5pt}
\section{Extension: requests of different processing times; replication}\label{sec:replica-placement}
Up to this point, we modeled a distributed request processing system, in which requests have the same size. In this section we show how our results extend to the model where the individual requests (constituting the load) have different durations and where the requests additionally have redundancy requirements (These extensions are particularly relevant for the problem of finding the replica placement in CDNs -- here different data pieces have different popularities and data redundancy is a common requirement for increasing the availability).

We introduce the following additional notation. A task is an individual request. $J_i = \{ J_{i}(k) \}$ denotes the set of tasks of organization $i$; $p_i(k)$ is the size (processing time) of the task $J_{i}(k)$.

First, let us analyze a problem in which the tasks have no redundancy requirements, i.e. each task has to be processed on exactly one server.

In order to find the optimal solution in this extended model, we start with solving the original problem (as defined in Section~\ref{sec:mathematical-model}) with $n_{i} = \sum_{k} p_i(k)$.
In order to derive the actual distribution of the tasks, we discretize the fractions $\rho_{ij}$ as follows. $i$ should relay to $j$ such subset $S_i(j) \subseteq J_i$ of its own tasks, so that the total error $\Sigma err(S_i(j))$ is minimized: 
\begin{align*}
err(S_i(j)) = |\sum_{k: J_{i}(k) \in S_i(j)} p_i(k) - \rho_{ij}n_{i}| \text{.}
\end{align*}

The rounding problem is the multiple subset problem with different knapsack capacities~\cite{Caprara2000111}. The problem is NP-complete but has a polynomial approximation algorithm.

Now consider a problem in which each organization must execute at least $R$ copies of each task; each copy of the task should be executed at a different location (the execution of the tasks is replicated). This setting models a CDN, but also job processing, where to increase survivability important jobs are replicated on different partions of a datacenter or on different datacenters.

In this extended problem we have to introduce additional constraint on the fractions $\rho_{ij}$ for the original problem (Section~\ref{sec:mathematical-model}): $\forall _{i, j} \: \rho_{ij} \leq \frac{1}{R}$, which guarantees that $R\rho_{ij} \leq 1$. With this constraint we can interpret $R\rho_{ij}$ as the probability of placing a copy of $J_{i}(k)$ at $j$; here the expected number of copies of $J_{i}(k)$ is $\sum_{j}R\rho{ij} = R$.

\vspace{-7pt}
\section{Related work}\label{sec:related-work}
The congestion games \cite{finiteCongestionGames, linearCongestionGames, routingGames, worstCaseEqulibria} define the model for analyzing the selfish behavior of users competing for commonly available resources. Similarly to our model, the cost of a particular resource is linearly proportional to the number of competitors using the resource. In contrast, our model more closely describes the cost of using a resource which depends also on the communication delay.

The assumptions in our model are similar as in the literature on network virualization~\cite{delay-load-server-alloc}. However in network virualization the problems regard locating services which is different from optimizing the quality of serving the common user requests. The complexity of the solutions depend on the number of configurations (which here is unbounded) thus the solutions cannot be applied to our model.

The continuos allocation of requests to servers in our model is analogous to the divisible load model~\cite{gallet09divisibleload} with constant-cost communication (a special case of the affine cost model~\cite{beaumont2005scheduling}) and multiple sources (multiple loads to be handled,~\cite{veeravalli2002efficient, drozdowski2008scheduling}). 
The main difference is the optimization goal: makespan is usually optimized in the divisible load model; in contrast, we optimize the average processing time, which, we believe, better models situations in which the load is composed of multiple, small requests issued by various users (the difference is analogous to $C_{\max}$ versus $\Sigma C_i$ debate in the classic multiprocessor job scheduling). The other difference is how the network topology is modelled. The divisible load theory typically studies datacenter-type systems, in which the network topology is known and is a limiting factor, thus the transmissions must be scheduled in a similar way to the computations. 

Distributed algorithms for load balancing mostly relay on local optimization techniques (see~\cite{Ackermann:2009:DAQ:1583991.1584046, Berenbrink:2011:DSL:2133036.2133152, Xu:1997:LBP:548748}). One of the most popular techniques is diffusive load balancing, similar in spirit to our distributed algorithm (see~\cite{conf/ipps/AdolphsB12} and the references inside for the current state of the art and \cite{Xu:1997:LBP:548748} for the basic mechanism description). These solutions, however, disregard the geographic distribution of the servers. Our algorithm uses different idea -- the diffusive process is used for calculating the relay fractions instead of for balancing the load. As the result, our local balancing must take into account different latencies between the servers which requires more subtle exchange mechanisms (Algorithms~\ref{alg:exchangingLoads}~and ~\ref{alg:distributedOptimal}).

Our game-theoretic approach is comparable to the selfish job model~\cite{vocking-selfishloadbalancing}: the jobs independently chose the processor on which to execute. While some studies consider mixed case equilibria (making the model continuous similarly to ours), our model considers also communication latency. The common infrastructure models tend to have a low price of anarchy (of order $\log m / \log \log m$~\cite{vocking-selfishloadbalancing}) --- the low price of anarchy in our model extends these results.

Content delivery networks are one of the motivations for our model. 
Large companies, like Akamai, specialize in delivering the content of their customers so that the end users experience the best quality of service. Akamai's architecture is based on DNS redirections \cite{draftingBehindAkamai, akamai2, akamaiPatent}. However, the  description of the algorithms optimizing replica placement and request handling are not disclosed. Still, Akamai's infrastructure is owned and controlled by a single entity (Akamai), thus they do not need to solve the game-theoretic equivalent of our model.

CoralCDN \cite{freedman2010experiences} is a p2p CDN consisting of users voluntarily devoting their bandwidth and storage to redistribute the content. In CoralCDN the popular content is replicated among multiple servers (which can be viewed as relaying the requests); the requests for content are relayed only between the servers with constrained pairwise RTTs (which ensures the proximity of delivering server). Our mathematical model formalizes the intuitions behind heuristics in CoralCDN.


\cite{dynamicReplicaPlacement} shows a CDN based on a DHT and heuristic algorithms to minimize the total processing time. Although each server has a fixed constrains on its load/bandwidth/storage capacity, the paper does not consider the relation between server load and its performance degradation. The evaluation is based on simulation; no theoretical results are included.

The problem of mirroring in the Internet is analyzed in \cite{staticReplicaPlacement1, staticReplicaPlacement2}. Both papers show different approaches to choosing locations for replicas so that the average network delay between data locations and end-users is minimized. The impact of servers' congestion is not taken into consideration. 

\vspace{-7pt}
\section{Conclusions}

In this paper we present and analyze a model of a distributed system that minimizes the request processing time by distributing the requests among servers. Existing models assume that the processing time is dominated either by the network communication delay or by congestion of servers. In contrast, \emph{in our model, the observed latency is the sum of the two delays: network delay and congestion}. Our model can be used in different kinds of problems in distributed systems, ranging from routing in content delivery networks to load balancing in a cloud of servers. 


We show that the problem of minimizing the total processing time can be stated as an optimization problem in the continuous domain. We prove that the problem is polynomially solvable; but, because of $O(m^6)$ complexity, standard solvers are not practical. We propose a distributed algorithm that, according to our experimental evaluation, even in a network consisting of thousands  of servers requires only a dozen of messages  sent by each server  to converge to a solution worse than at most 0.1\% of the optimum (not counting the gossiping to exchange the information). We show how to estimate the distance between the current solution found by the algorithm and the optimal solution. The estimation is difficult in practice, as it requires solving the subproblem of finding the maximal flow of the minimal cost in a graph. However, the distributed algorithm still outperforms standard optimization techniques. Based on the experiments, we argue that in practice this part of the algorithm can be omitted, as it does not influence the algorithm efficiency.

We also analyze how the lack of coordination influences the total processing time. We give theoretical bounds for the price of anarchy for homogeneous networks and high average loads. Additionally, we assess the price of anarchy experimentally on heterogenous networks. In both cases the price of anarchy is low ($1 + \frac{2cs}{l_{av}}$ in the theoretical analysis, and below $1.15$ in the experiments). 

Our results --- the low price of anarchy and an efficient distributed optimization algorithm --- indicate that a fully distributed query processing system can be efficient. Thus, instead of buying services from dedicated cloud providers or CDN operators, smaller organizations, such as ISPs or universities, can gather in consortia effectively serving the participators' needs.


\vspace{-7pt}
\bibliographystyle{abbrv}
\bibliography{contentDelivery}

\begin{thebibliography}{10}

\bibitem{Ackermann:2009:DAQ:1583991.1584046}
H.~Ackermann, S.~Fischer, M.~Hoefer, and M.~Sch\"{o}ngens.
\newblock Distributed algorithms for qos load balancing.
\newblock In {\em Proceedings of the twenty-first annual symposium on
  Parallelism in algorithms and architectures}, SPAA '09, pages 197--203, 2009.

\bibitem{conf/ipps/AdolphsB12}
C.~P.~J. Adolphs and P.~Berenbrink.
\newblock Improved bounds for discrete diffusive load balancing.
\newblock In {\em IPDPS}, pages 820--826. IEEE Computer Society, 2012.

\bibitem{allard2007sofa}
J.~Allard, S.~Cotin, F.~Faure, P.~Bensoussan, F.~Poyer, C.~Duriez,
  H.~Delingette, L.~Grisoni, et~al.
\newblock Sofa-an open source framework for medical simulation.
\newblock In {\em Medicine Meets Virtual Reality, MMVR 15}, 2007.

\bibitem{delay-load-server-alloc}
D.~Arora, A.~Feldmann, G.~Schaffrath, and S.~Schmid.
\newblock On the benefit of virtualization: Strategies for flexible server
  allocation.
\newblock In {\em Proceedings of USENIX Workshop on Hot Topics in Management of
  Internet, Cloud, and Enterprise Networks and Services (Hot-ICE '11)}, 2011.

\bibitem{beaumont2005scheduling}
O.~Beaumont, H.~Casanova, A.~Legrand, Y.~Robert, and Y.~Yang.
\newblock Scheduling divisible loads on star and tree networks: results and
  open problems.
\newblock {\em Parallel and Distributed Systems, IEEE Transactions on},
  16(3):207--218, 2005.

\bibitem{Berenbrink:2011:DSL:2133036.2133152}
P.~Berenbrink, M.~Hoefer, and T.~Sauerwald.
\newblock Distributed selfish load balancing on networks.
\newblock In {\em Proceedings of the Twenty-Second Annual ACM-SIAM Symposium on
  Discrete Algorithms}, SODA '11, pages 1487--1497, 2011.

\bibitem{auctionBasedMinCostFlow}
D.~P. Bertsekas.
\newblock Auction algorithms for network flow problems: A tutorial
  introduction.
\newblock {\em Computational Optimization and Applications}, 1:7--66, 1992.

\bibitem{Caprara2000111}
A.~Caprara, H.~Kellerer, and U.~Pferschy.
\newblock A {PTAS} for the multiple subset sum problem with different knapsack
  capacities.
\newblock {\em Information Processing Letters}, 73(3-4), 2000.

\bibitem{Chan-TinH11}
E.~Chan-Tin and N.~Hopper.
\newblock Accurate and provably secure latency estimation with treeple.
\newblock In {\em NDSS}. The Internet Society, 2011.

\bibitem{Chawla:2011:SCS:2002181.2002214}
A.~Chawla, B.~Reed, K.~Juhnke, and G.~Syed.
\newblock Semantics of caching with spoca: a stateless, proportional,
  optimally-consistent addressing algorithm.
\newblock In {\em USENIXATC}, 2011.

\bibitem{dynamicReplicaPlacement}
Y.~Chen, R.~H. Katz, and J.~Kubiatowicz.
\newblock Dynamic replica placement for scalable content delivery.
\newblock In {\em IPTPS, Proceedigs}, pages 306--318, London, UK, 2002.
  Springer-Verlag.

\bibitem{finiteCongestionGames}
G.~Christodoulou and E.~Koutsoupias.
\newblock The price of anarchy of finite congestion games.
\newblock In {\em STOC, Proceedings}, pages 67--73, 2005.

\bibitem{staticReplicaPlacement2}
E.~Cronin, S.~Jamin, C.~Jin, A.~R. Kurc, D.~Raz, Y.~Shavitt, and S.~Member.
\newblock Constrained mirror placement on the internet.
\newblock In {\em JSAC}, pages 31--40, 2002.

\bibitem{drozdowski2008scheduling}
M.~Drozdowski and M.~Lawenda.
\newblock Scheduling multiple divisible loads in homogeneous star systems.
\newblock {\em Journal of Scheduling}, 11(5):347--356, 2008.

\bibitem{dutot2004bi}
P.~Dutot, L.~Eyraud, G.~Mouni{\'e}, and D.~Trystram.
\newblock Bi-criteria algorithm for scheduling jobs on cluster platforms.
\newblock In {\em IPDPS, Proc.}, pages 125--132. ACM, 2004.

\bibitem{freedman2010experiences}
M.~Freedman.
\newblock Experiences with coralcdn: A five-year operational view.
\newblock In {\em NSDI USENIX, Proceedings}, 2010.

\bibitem{gallet09divisibleload}
M.~Gallet, Y.~Robert, and F.~Vivien.
\newblock Divisible load scheduling.
\newblock In Y.~Robert and F.~Vivien, editors, {\em Introduction to
  Scheduling}. CRC Press, Inc., 2009.

\bibitem{minimumCirculation}
A.~V. Goldberg and R.~E. Tarjan.
\newblock Finding minimum-cost circulations by successive approximation.
\newblock {\em Math. Oper. Res.}, 15:430--466, July 1990.

\bibitem{bestConvex}
D.~Goldfarb and S.~Liu.
\newblock An o(n3) primal interior point algorithm for convex quadratic
  programming.
\newblock {\em Math. Program.}, 49:325--340, January 1991.

\bibitem{convexComplexity}
D.~S. Hochbaum.
\newblock Complexity and algorithms for nonlinear optimization problems.
\newblock {\em Annals OR}, 153(1):257--296, 2007.

\bibitem{worstCaseEqulibria}
E.~Koutsoupias and C.~H. Papadimitriou.
\newblock Worst-case equilibria.
\newblock {\em Computer Science Review}, 3(2):65--69, 2009.

\bibitem{quadraticSolvability}
M.~K. Kozlov, S.~P. Tarasov, and L.~G. Khachiyan.
\newblock {Polynomial Solvability of Convex Quadratic Programming}.
\newblock {\em Doklady Akademiia Nauk SSSR}, 248, 1979.

\bibitem{akamaiPatent}
F.~Leighton and D.~Lewin.
\newblock Global hosting system.
\newblock US Patent No. 6,108,703.

\bibitem{akamai2}
R.~Mahajan.
\newblock How akamai works?
\newblock
  \url{http://research.microsoft.com/en-us/um/people/ratul/akamai.html}.

\bibitem{routingGames}
N.~Nisan, T.~Roughgarden, E.~Tardos, and V.~V. Vazirani.
\newblock {\em Algorithmic Game Theory}, chapter Routing Games.
\newblock Cambridge University Press, 2007.

\bibitem{Nygren:2010:ANP:1842733.1842736}
E.~Nygren, R.~K. Sitaraman, and J.~Sun.
\newblock The {Akamai} network: a platform for high-performance internet
  applications.
\newblock {\em SIGOPS Oper. Syst. Rev.}, 44:2--19, August 2010.

\bibitem{pallis2006content}
G.~Pallis and A.~Vakali.
\newblock Content delivery networks.
\newblock {\em Communications of the ACM}, 49(1):101, 2006.

\bibitem{staticReplicaPlacement1}
L.~Qiu, V.~N. Padmanabhan, and G.~M. Voelker.
\newblock On the placement of web server replicas.
\newblock In {\em IEEE INFOCOM, Proceedings}, pages 1587--1596, 2001.

\bibitem{linearCongestionGames}
A.~Roth.
\newblock The price of malice in linear congestion games.
\newblock In {\em WINE, Proceedings}, pages 118--125, 2008.

\bibitem{Rzadca2010ReplicaPlacementin}
K.~Rzadca, A.~Datta, and S.~Buchegger.
\newblock Replica placement in p2p storage: Complexity and game theoretic
  analyses.
\newblock In {\em ICDCS, Proceedings}, pages 599--609, 2010.

\bibitem{draftingBehindAkamai}
A.-J. Su, D.~R. Choffnes, A.~Kuzmanovic, and F.~E. Bustamante.
\newblock Drafting behind {Akamai}.
\newblock {\em SIGCOMM Comput. Commun. Rev.}, 36:435--446, August 2006.

\bibitem{Szymaniak04scalablecooperative}
M.~Szymaniak, G.~Pierre, and M.~Steen.
\newblock Scalable cooperative latency estimation.
\newblock In {\em ICPADS}, 2004.

\bibitem{veeravalli2002efficient}
B.~Veeravalli and G.~Barlas.
\newblock Efficient scheduling strategies for processing multiple divisible
  loads on bus networks.
\newblock {\em Journal of Parallel and Distributed Computing}, 62(1):132--151,
  2002.

\bibitem{vocking-selfishloadbalancing}
B.~Vocking.
\newblock Selfish load balancing.
\newblock In N.~Nisan, T.~Roughgarden, E.~Tardos, and V.~V. Vazirani, editors,
  {\em Algorithmic Game Theory}. Cambridge University Press, 2007.

\bibitem{Xu:1997:LBP:548748}
C.~Xu and F.~C. Lau.
\newblock {\em Load Balancing in Parallel Computers: Theory and Practice}.
\newblock Kluw. Acad. Pub., 1997.

\end{thebibliography}

\vspace{-10pt}
\appendix[Removing Negative Cycles]\label{sec:remov-negat-cycl}

The problem of negative cycles removal can be reduced to finding the maximal flow of the minimal cost in a graph. The problem of finding the maximal flow of the minimum cost is well studied in the literature;
In particular the auction algorithms \cite{auctionBasedMinCostFlow}
and the approximation method for finding minimum circulation  \cite{minimumCirculation} are the examples of the distributed algorithms solving the problem.

For the purpose of proving the reduction we introduce the following notation. $out(\rho', i)$ denotes the total amount of requests that in a partial solution $\rho'$ are relied by a server $i$ to all other servers: $out(\rho', i) = \sum_{j \neq i} r_{ij}$. $in(\rho', i)$ denotes the total amount of requests that in $\rho'$ are relied by all other servers to $i$, $in(\rho', i) = \sum_{j \neq i} r_{ji}$.

For each server $i$ we introduce two graph vertices: the front $i_{f}$ and the back $i_{b}$.
There are two additional vertices: $s$ (source) and $t$ (target).
The source $s$ is linked with each front node, $i_{f}$ with an edge $(s, i_f)$ with zero cost and capacity equal to
$out(\rho', i)$. Each back node, $i_b$ is linked with the target $t$ with an edge $(i_b, t)$ with zero cost and capacity equal to $in(\rho', i_{b})$.

There are also edges between front and back nodes: for each pair $(i_{f}, j_{b})$, $i \neq j$ there is an edge with cost equal to $c_{ij}$ and infinite capacity. 

The maximal flow of the minimal cost $f$ between $s$ and $t$ can be mapped to a new partial solution $\rho''$: a flow  on an edge ($i_{f}$, $j_{b}$) $f_{ij}$ corresponds to server $i$ relying $f_{ij}$ of its own requests to server $j$. 
Observe that, as capacity $(s, i_f) = out(\rho', i)$  and capacity $(i_b, t) = in(\rho', i)$,
the load of $i$-th server in $\rho''$ is equal to its load in $\rho’$.

\vspace{-7pt}
\appendix[Validation of the constant latency]

We experimentally verified how the amount of the load sent over the network influences the communication delay between the servers. We randomly selected 60 PlanetLab servers, scattered around Europe, and simulated different intensity of the background load in the following way. Each server choses its 5 neighbors randomly
Then the servers start sending data with constant throughput to its 5 neighbors. In different experiments, we used 8 values of the throughputs: 10KB/s, 20KB/s, 50KB/s, 100KB/s, 200KB/s, 500KB/s, 1MB/s, 2MB/s. If a particular throughput was not achievable, the server was just sending data with the maximal achievable throughput. 
For each value of the background load we calculated the average round trip time (RTT) between the server and each of its 5 neighbors (we used the average from 300 RTT samples).

Let $rtt(s_i, s_j, t_b)$ denote the average rtt between servers $s_i$ and $s_j$ with the background load generated with throughput $t_b$. For each pair of the servers $s_{i}$ and $s_j$ for which we measured the RTT, and for each value of the background throughput $t_{b}$ we calculated the relative deviation of the average throughput caused by the increase of the background load compared to the minimal throughput 10KB/s: $e(s_i, s_j, b_t) = \frac{rtt(s_i, s_j, t_b) - rtt(s_i, s_j, 10KB/s)}{rtt(s_i, s_j, 10KB/s)}$. For each value of the background throughput, we removed 5\% of the largest deviations and then calculated the mean from deviations $e(s_i, s_j, b_t)$, averaged over all pairs of servers ($\mu$). For each value of the background throughput we additionally calculated the standard deviations ($\sigma$). These results are presented in Table \ref{table:constantRTT}.

\begin{table}[t]
\parbox{.45\linewidth}{
\centering
\begin{tabular}{|c|c|c|}
\cline{1-3}
\multirow{2}{*}{$t_b$} & \multicolumn{2}{c|}{$e(\cdot, \cdot, t_b)$} \\
\cline{2-3}
& $\mu$ & $\sigma$ \\
\cline{1-3}
& & \\
10 KB/s & 0.0 & 0.0 \\
20 KB/s & -0.05 & 0.21 \\
50 KB/s & -0.05 & 0.27 \\
0.1 MB/s & -0.08 & 0.33 \\
\cline{1-3}
\end{tabular}
}
\hfill
\parbox{.45\linewidth}{
\begin{tabular}{|c|c|c|}
\cline{1-3}
\multirow{2}{*}{$t_b$} & \multicolumn{2}{c|}{$e(\cdot, \cdot, t_b)$} \\
\cline{2-3}
& $\mu$ & $\sigma$ \\
\cline{1-3}
& & \\
0.2 MB/s & 0.0 & 0.37 \\
0.5 MB/s & 0.28 & 0.8 \\
2 MB/s & 0.45 & 1.31 \\
5 MB/s & 0.18 & 0.8 \\
\cline{1-3}
\end{tabular}
}
\caption{The relative deviation of the average throughput caused by the increase of the background load (after removal of 5\% largest deviations).}
\label{table:constantRTT}
\vspace{-10pt}
\end{table}

From the data we see that up to $b_t = 0.2$MB/s, which corresponds to the case where each server accepts $5 \cdot 0.2 \cdot 8 = 8$Mb/s of incoming data, the average RTT was not influenced by the background throughput. This is also confirmed by the statistical analysis of the data run for the RTTs (instead of for deviations). For $b_t \leq 0.2$MB/s the ANOVA test (which we run for the whole population -- without removing 5\% of the highest RTTs) confirmed the
lack of dependency
for over $56\%$ of the pairs of servers. For $b_t \leq 0.1$MB/s (corresponding to 4Mb/s of incoming throughput) the ANOVA test confirmed null hypothesis for over 70\% of the pairs of servers and for $b_t \leq 50$KB/s
for over 90\% of the pairs. We consider that these results strongly justify the assumption of a constant latency in our model.

\end{document}